\documentclass[11pt]{amsart}

\usepackage{amscd}
\usepackage{amssymb}
\usepackage{amsthm}
\usepackage{amsmath}

\theoremstyle{plain}

\newtheorem{Prop}{Proposition}

\theoremstyle{remark}

\numberwithin{equation}{section}

\usepackage{hyperref}

\begin{document}

\title{L\"{u}{}ders Theorem for Unsharp Quantum Measurements}
\author[P Busch and J Singh]{Paul Busch and Javed Singh\\
$\text{\rm\tiny Former  Affiliation (PB): Department of Mathematics, University of Hull, UK}$\\
$\text{\rm\tiny Current (PB): Department of Mathematics, University of York, UK}$\\
$\text{\rm\tiny E-mail address: {\tt paul.busch@york.ac.uk}}$\\
\hspace{1cm}\hfill\\
$\text{\rm\tiny Published in: {\em Physics Letters A} {\bf 249} 10--12 (1998).}$\\
DOI: \href{http://dx.doi.org/10.1016/S0375-9601(98)00704-X}{10.1016/S0375-9601(98)00704-X}.
}

\maketitle

\begin{abstract}
A theorem of L\"{u}{}ders states that an ideal measurement of a sharp
discrete observable does not alter the statistics of another sharp
observable if, and only if, the two observables commute. It will be shown
that this statement extends to certain pairs of unsharp observables.
Implications for local relativistic quantum theory will be discussed.
\end{abstract}

\section{\protect\smallskip Introduction}

\vspace{0in}\vspace{-0.03in}\smallskip \vspace{0in}Let ${\mathcal H}$ be a
complex separable Hilbert space, $A$ a self-adjoint operator with\vspace{%
-0.03in} discrete spectrum and spectral decomposition $A=\sum_ia_iE_i^A$.
Here the $a_i,i=1,2,\cdots ,N\le \infty $, are the distinct eigenvalues,
with spectral projections $E_i^A$ adding up to unity, $\sum_iE_i^A=I$. A
non-selective ideal measurement of the observable represented by $A$ gives
rise to the L\"{u}{}ders transformation of states (density operators) \cite
{lud} 
\begin{equation}
\rho \mapsto {\mathcal I}_L\left( \rho \right) =\sum_{i=1}^NE_i^A\rho \,E_i^A.
\label{lu1}
\end{equation}
Let $B$ be any self-adjoint operator on ${\mathcal H}$. Then the L\"{u}{}ders
theorem asserts that 
\begin{equation}
{\rm tr}[{\mathcal I}_L(\rho )\cdot \,B]={\rm tr}[\rho \cdot \,B]  \label{lu2}
\end{equation}
holds for all states $\rho $ exactly when $B$ commutes with all $E_i$ and
thus with $A$. This result is interpreted as follows: if $A$ and $B$
commute, then the outcomes of a measurement of an observable represented by $%
B$ do not depend on whether $A$ has been measured first.

In this paper the question will be investigated as to what extent the
L\"{u}ders theorem pertains also to {\sl unsharp} observables, represented
by {\sc pov} measures that are not {\sc pv} measures. If the observable
measured first, $A$, is unsharp, it will be represented by a complete set of
coexistent effects $E_i,i=1,\cdots ,N$, $\sum_iE_i=I$. In this case the
L\"{u}{}ders transformation is defined as \cite{blm} 
\begin{equation}
\rho \mapsto {\mathcal I}_L(\rho )=\sum_{i=1}^NE_i^{1/2}\rho \,E_i^{1/2}.
\label{lu3}
\end{equation}
The second observable could also be an unsharp observable. In this case, $B$
will represent an effect in the range of the corresponding {\sc pov}
measure. The question is whether commutativity is still necessary for (\ref
{lu2}) to hold for all states.

\section{Generalised L\"{u}ders Theorem}

We will formulate and prove two propositions which cover two classes of
cases of a general L\"{u}{}ders theorem yet to be proved for arbitrary {\sc %
pov} measures.

\noindent 

\begin{Prop}
Let $E_i,$ $i=1,2,\cdots ,N\le \infty $, be a complete family of effects.
Let $B=\sum_kb_kP_k$ be an effect with discrete spectrum given by the
strictly decreasing sequence of eigenvalues $b_k\in [0,1]$ and spectral
projections $P_k,$ $k=1,\cdots ,K\le \infty $. Then ${\rm tr}[{\mathcal I}%
_L(\rho )\cdot \,B]={\rm tr}[\rho \cdot \,B]$ [Eq. ~(\ref{lu2})] holds for
all states $\rho $ exactly when all $E_i$ commute with $B$.
\end{Prop}

\begin{proof}
Commutativity is obviously sufficient for (\ref{lu2}) to hold. To prove the
converse implication, assume that (\ref{lu2}) holds for all states $\rho $.
This is equivalent to the following equation : 
\begin{equation}
B=\sum_{i=1}^NE_i^{1/2}B\,E_i^{1/2}.  \label{p1}
\end{equation}
Let $\varphi \in {\mathcal H}$ be an arbitrary vector. Then Eq. (\ref{p1}) gives:

\begin{eqnarray}
\langle P_1\varphi |BP_1\varphi \rangle &=&b_1\langle P_1\varphi |P_1\varphi
\rangle =\sum_{i=1}^N\langle E_i^{1/2}P_1\varphi |BE_i^{1/2}P_1\varphi
\rangle  \nonumber \\
&\le &b_1\sum_{i=1}^N\langle E_i^{1/2}P_1\varphi |E_i^{1/2}P_1\varphi
\rangle =b_1\langle P_1\varphi |P_1\varphi \rangle .  \label{p2}
\end{eqnarray}
It follows that equality must hold and thus all terms of the first sum must
equal the corresponding terms of the second sum. Taking into account the
fact that $b_1=||B||$, this can be expressed as follows: $%
||(b_1I-B)^{1/2}E_i^{1/2}P_1\varphi ||=0$, which is to hold for all $\varphi 
$. Therefore, $BE_i^{1/2}P_1=b_1E_i^{1/2}P_1$, and so $%
P_1E_i^{1/2}P_1=E_i^{1/2}P_1$. This implies that all $E_{i\text{ }}$commute
with $P_1$.

Now proceed as follows: substitute $B_1=B$ with $B_2=B-b_1P_1$. The
commutativity just proven together with Eq. (\ref{p1}) entail the same
equation for $B_2$: 
\begin{equation}
B_2=\sum_{i=1}^NE_i^{1/2}B_2\,E_i^{1/2}  \label{p3}
\end{equation}
Applying the same argument as before with $b_2=||B_2||$ yields the
commutativity of all $E_{i\text{ }}$with $P_2$. Thus one concludes
inductively that the $E_i$ commute with {\sl all }$P_k$ and hence with $B$. 
\end{proof}

\noindent 

\begin{Prop}
Let $E_1=E$ be an effect, $E_2=I-E$. Let $B$ be any effect. Then ${\rm tr}[%
{\mathcal I}_L(\rho )\cdot \,B]={\rm tr}[\rho \cdot \,B]$ [Eq. ~(\ref{lu2})]
holds for all states $\rho $ exactly when $E_1$ commutes with $B$.
\end{Prop}

\begin{proof}
Eq. ~(\ref{lu2}), taken for all states $\rho $, is equivalent to 
Eq. ~(\ref{p1}%
) [with $N=2$]. If $E_1$ commutes with $B$, then (\ref{p1}) follows
trivially. Conversely, assume this equation holds. By a simple algebraic
manipulation one deduces the following: 
\[
E_1B+BE_1=2E_1^{1/2}BE_1^{1/2},
\]
and this is equivalent to 
\[
\left[ E_1^{1/2},\left[ E_1^{1/2},B\right] \right] =0,
\]
where $\left[ A,B\right] $ denotes the commutator of the bounded 
operators $A,B$. It follows that the self-adjoint operator 
$C:=i\left[E_1^{1/2},B\right] $ is quasi-nilpotent, 
i.e. ~its spectrum is $\{0\}$. [This
follows from  a Lemma stated and proved in the Appendix.]  
Therefore $C=0$, and thus $E_1$ and $B$ commute. 
\end{proof}

\section{Discussion}

We have generalised the L\"{u}ders theorem to two classes of unsharp
measurements: in the first case, arbitrary L\"{u}ders transformations are
allowed while the test effect $B$ has discrete spectrum; in the second case,
the spectrum of the effect $B$ is arbitrary but the class of  L\"{u}ders
transformations is restricted to those corresponding to simple observables
[i.e. ~{\sc pov }measures with ranges $\{E_1,I-E_1\}$]. Note that Proposition
2 holds not only for effects $B$ but for any bounded self-adjoint operator.

We take these results as indications that the statement of the L\"{u}ders
theorem can be expected to hold in full generality. Unfortunately, the above
simple proofs do not extend in an obvious way to the general case so that
further investigations are required. From a physical point of view, however,
the present results may already be considered sufficient for the following
considerations. If an ideal measurement of an observable defined by the
complete set of effects $E_1,E_2,\cdots $ is realisable, then it seems
plausible that it should also be possible to perform ideal measurements of
the simple observables $E_i,I-E_i$. Thus Proposition 2 would apply to each
of those. 

Now consider the well-known application of the L\"{u}ders theorem in the
context of relativistic quantum theory. Here the observable $A$ defined by
the effects $E_1,E_2,\cdots $ is taken to belong to a local algebra
associated with some bounded spacetime region $X$, and the effect $B$ is
considered to belong to a local algebra associated with another bounded
region  $Y$ with a spacelike separation from the first region. The
requirement of (Einstein) causality then states that a measurement of any
observable $A$ in region $X$ should not have an observable effect in region $%
Y$. This condition is formalised by means of equations (\ref{lu2}) or (\ref
{p1}) if the measurement of $A$ is assumed to be an ideal measurement with
the ensuing non-selective L\"{u}ders transformation (\ref{lu1},\ref{lu3}).
If $A$ is a simple observable, then Proposition 2 applies and we conclude
that every effect in the algebra associated with region $Y$ commutes with
the effects in the range of $A$. Hence, local commutativity follows from
Einstein causality under the weaker assumption of the observables being
defined as {\sc pov} measures rather than self-adjoint operators. This
generalises a result of a (fundamental but apparently not too well known)
paper by Schlieder \cite{schl}.

\appendix
\section{A Lemma}

We present a proof of the following fact: Let $d:{\mathcal A}\rightarrow {\mathcal A}
$ be a bounded derivation on the unital $C^{*}$-algebra ${\mathcal A}$ with unit 
$I$. Let $a\in {\mathcal A}$ be such that $d^2a=0$. Then $da$ is
quasi-nilpotent. 

First note that $d^2a=0$ implies $d^ka=0$ for $k=2,3,\cdots $. This can be
used to prove by induction that $d^n(a^n)=n!(da)^n$, $n=1,2,\cdots $. It
then follows that 
\[
||(da)^n||^{1/n}=(n!)^{-1/n}||d^n(a^n)||^{1/n}\le (n!)^{-1/n}||d||~||a||.
\]
Therefore the spectral radius of $da$, $r(da)$, is
\[
r(da)=\lim_{n\rightarrow \infty }||(da)^n||^{1/n}=0,
\]
so that the spectrum of $da$ is $\{0\}$.

The proof of Proposition 2 makes use of the fact that $d:B\rightarrow \left[
E_1^{1/2},B\right] $ is a bounded derivation on the algebra of bounded
operators.



\begin{thebibliography}{9}
\bibitem{lud}   G. ~L\"{u}ders: \"{U}ber die Zustands\"{a}nderung durch den
Messprozess. {\em Ann. ~Physik} (6.F.) {\bf 8 } (1951) 322--328.

\bibitem{blm}   P. ~Busch, P. ~Lahti, P. ~Mittelstaedt:  {\em The Quantum
Theory of Measurement}, Springer, Berlin, 2nd revised ed., 1996.

\bibitem{schl}  S. ~Schlieder: Einige Bemerkungen zur Zustands\"{a}nderung
von relativistischen quantenmechanischen Systemen durch Messungen und zur
Lokalit\"{a}tsforderung. {\em Commun. ~math. ~Phys. }{\bf 7} (1968) 305--331.
\end{thebibliography}
\end{document}